\documentclass[11pt]{article} 
\usepackage{csquotes}
\usepackage{tikz}
\usepackage{amsthm,amsmath,amssymb}
\usepackage{graphicx}
\usepackage[colorlinks=true,citecolor=black,linkcolor=black,urlcolor=blue]{hyperref}
\usepackage[linesnumbered,ruled,vlined]{algorithm2e}
\usepackage{dsfont}
\usepackage[top=2.5cm, bottom=3cm, left=2.75cm, right=2.75cm]{geometry}
\usepackage{authblk}
\theoremstyle{plain}
\newtheorem{theorem}{Theorem}
\newtheorem{lemma}[theorem]{Lemma}
\newtheorem{corollary}[theorem]{Corollary}

\theoremstyle{definition}
\newtheorem{example}[theorem]{Example}

\title{The Exact Query Complexity of Yes-No Permutation Mastermind}
\author[1]{Mourad El Ouali}
\author[2]{Volkmar Sauerland}
\affil[1]{Polydisciplinary Faculty Ouarzazate, University Ibn Zohr, Agadir, Morocco; elouali@math.uni-kiel.de}
\affil[2]{Department of Mathematics, Kiel University, Kiel, Germany; sauerland@math.uni-kiel.de}
\date{}

\begin{document}

\maketitle

\begin{abstract}
Mastermind is famous two-players game.
The first player (\emph{codemaker}) chooses a secret code
which the second player (\emph{codebreaker}) is supposed to crack
within a minimum number of code guesses (queries).
Therefore, codemaker's duty is to help codebreaker by providing
a well-defined error measure between the secret code and the guessed code
after each query.
We consider a variant, called Yes-No AB-Mastermind,
where both secret code and queries must be repetition-free
and the provided information by codemaker only indicates
if a query contains any correct position at all.
For this Mastermind version with $n$ positions and $k\ge n$ colors
and $\ell:=k+1-n$ we prove a lower bound of $\sum_{j=\ell}^{k}\log_2 j$
and an upper bound of $n\log_2 n+k$ on the number of 
queries necessary to break the secret code.
For the important case $k=n$,
where both secret code and queries represent permutations,
our results imply
an exact asymptotic complexity of $\Theta(n \log n)$ queries.
\end{abstract}

\section{Introduction}
Mastermind is a popular two-player board game invented in 1970
by Mordecai Meirowitz.
Its idea is that a codemaker chooses a secret code of fixed length $n$,
where each position is selected from a set of $k$ colors.
The second player, codebreaker, has to identify the secret code
by a finite sequence of corresponding code guesses (queries),
each of which is replied with the number of matching positions
and the number of further correct colors.
The original game is played by picking pegs of $k=6$ different colors
and placing them into rows with $n=4$ holes,
where the number of rows (allowed queries) for codebreaker is limited.

Generalizing the situation to arbitrary many positions and colors,
codemaker selects a vector $y\in {[k]}^n$
and codebreaker gives in each iteration a query
in form of a vector $x\in {[k]}^n$.
In the original setting,
codemaker's reply is the so called \emph{black-white} error measure,
consisting of a pair of numbers,
where the first number, $\mathrm{black}(x,\,y)$,
is the number of positions in which $x$ and $y$ coincide
and the second number, $\mathrm{white}(x,\,y)$,
is the number of additional colors which appear in both x and y,
but at different positions.
In this paper, we consider a variant, called Yes-No AB-Mastermind,
which is defined by the following properties
\begin{itemize}
\item Both secret code and queries must be repetition-free.
This property is indicated by the prefix AB
and stems from the AB game, better known as \enquote{Bulls and Cows},
which was even known prior to the commercial version of Mastermind
with color repetitions.
\item The provided information by codemaker only answers the question
whether or not a query contains any correct position at all.
This property is introduced by us and referred to by the term \enquote{Yes-No}.
\end{itemize}

{\bf Related Works:}
Mastermind and its variants have been analyzed under different aspects.
One of the first analysis of the commercial version
with $n=4$ positions and $k=6$ colors is due to Knuth~\cite{K77}
and shows that each code can be cracked in at most $5$ queries.
Even before the appearance of Mastermind as a commercial game
Erd\H{o}s and R\'enyi~\cite{ER63} analyzed the asymptotic query complexity
of a similar problem with two colors in 1963.
After Knuth's analysis of the commercial game,
many different variants of Mastermind
with arbitrary code length $n$ and number of colors $k$
have been investigated.
For example, Black-Peg Mastermind restricts its error measure
between two codes $x$ and $y$ to the single value $\mathrm{black}(x,y)$
(i.e., the exact number of positions where both codes coincide).
This version was introduced in 1983 by Chv\'atal~\cite{C83} for the case $k=n$,
who provides a deterministic adaptive strategy
using $2n\lceil\log_2 k\rceil + 4n$ queries.
Improved upper bounds for this variant and arbitrary $n$ and $k$
where given by Goodrich~\cite{Goo09b}
and later by J\"ager and Peczarski~\cite{JP11}
but remained in the order of $\mathcal{O}(n\log_2 n)$.
Doerr et al.~\cite{DDST16} provided a randomized codebreaker strategy
that only needs $\mathcal{O}(n \log\log n)$ queries in expectation
and showed that this asymptotic order even holds
for up to $n^2\log\log n$ colors,
if both black and white information is allowed.
A first upper bound for AB Mastermind
was given by Ker-I Ko and Shia-Chung Teng~\cite{KT86} for the case $k=n$,
i.e., secret code and queries represent permutations of $[n]$.
Their non-constructive strategy yields an upper bound of $2n\log_2 n + 7n$
queries.
A constructive strategy by El Ouali and Sauerland~\cite{ES13}
reduced this upper bound by a factor of almost $2$
and also included the case $k>n$ of \emph{Black-Peg} AB-Mastermind.
The term Black-Peg labels the situation that the error measure
between secret code and queries is only the \enquote{black} information,
i.e.\ the number of coinciding positions,
while the \enquote{white} information (see above) is omitted.
El Ouali et al.~\cite{EGSS18} combined the upper bound
with a lower bound of $n$ queries,
which is implied by a codemaker (cheating) strategy.
It improved the lower bound of $n-\log\log n$
by Berger et al~\cite{BCS16}.
However, a gap between $\Omega(n)$ and $\mathcal{O}(n\log_2 n)$
remains for this Mastermind variant.
Some facts indicate that closing this gap means to improve both bounds.
On the one hand, a careful consideration of the partition of the remaining
searchspace with respect to all possible codemaker replies might yield
a refined codemaker (cheating) strategy and possibly increase the lower bound.
On the other hand, overcoming the sequential learning process of the
codebreaker's binary search strategy might decrease the upper bound.
The latter presumption is reinforced by results of
Afshani et al.~\cite{AAD+19},
who consider another permutation-based variant of Mastermind.
There, the secret code is a combination of a binary string and a permutation,
(both of length $n$), queries are binary strings of length $n$,
and the error measure returns the number of leading coincidences
in the binary string with respect to the order of the permutation.
For this setting, which is also a generalization of the
popular \emph{leading ones} test problem in black box optimization,
the authors prove an exact asymptotic query complexity of $\Theta(n\log n)$
for deterministic strategies but a randomized query
complexity of  $\Theta(n\log\log n)$.

One of the ultimate goals in the analysis of Mastermind variants
is to prove the exact asymptotic query complexity.
As mentioned above, closing the asymptotic gap
between the lower $\Omega(n)$ bound
and the upper $\mathcal{O}(n\log_2 n)$ bound
is an unsolved problem for Black-Peg AB Mastermind.
A related open question is whether the same asymptotic
number of queries is required for both,
(Black-Peg) Mastermind with color repetition and
(Black-Peg) AB Mastermind.

{\bf Our Contribution:}
We consider a new variant of AB-Mastermind
which is more difficult to play for codebreaker
since the error-measure provided by codemaker is
less informative.
Here, for a secret code $y$ the answer $\mathrm{info}(\sigma,y)$
to a query $\sigma$ is \enquote{yes}
if some of its positions coincide with the secret code,
otherwise the answer is \enquote{no}.
We first analyze the worst-case performance of query strategies
for this Mastermind variant
and give a lower bound of $\sum_{j=\ell}^{k}\log_2 j$ queries for $k\ge n$,
which becomes $n\log_2 n-n$ in the case $k=n$.
The lower bound even holds
if codebreaker is allowed to use repeated colors in his queries.
We further present a deterministic polynomial-time algorithm
that identifies the secret code.
This algorithm is a modification of the constructive strategy
of El Ouali et al.~\cite{EGSS18}.
It returns the secret code in at most $(n-3)\log_2 n+\frac{5}{2}n-1$ queries
in the case $k=n$
and in less than $(n-2)\log_2 n+k+1$ queries in the case $k>n$.
For the important case $k=n$, our results imply the exact
asymptotic query complexity of $\Theta(n\log_2 n)$.
Since the considered \enquote{Yes-No} error measure
implies a new variant of AB-Mastermind,
there is no previous reference to compare our results to. 
 
\section{Results}
\subsection{Lower Bound on the Number of Queries}\label{sec:lb}
  To simulate the worst case, we allow codemaker to \enquote{cheat}
  in a way that after every query he may decide for a new secret code
  concerning the answers given so far.
\begin{theorem}
  Let $k,n\in\mathbb{N}$, $k\ge n$ and $\ell:=k+1-n$.
  Every strategy for Yes-No AB-Mastermind needs at least
  $
    \sum\limits_{j=\ell}^{k}\log j
  $
  queries in the worst case.
\end{theorem}
\begin{proof}
  We give a codemaker (cheating) strategy that implies the lower bound.
  For $i\in\mathds{N}$ let $M_i$ denote the set of secrets
  that are still possible after the $i$-th query has been answered,
  starting with $M_0:=\{y\in[k]^n\,|\, \forall i\ne j\in[n]: y_i\ne y_j\}$.
  Let $M_i^{\mathrm{yes}}\subset M_i$ be the set of secrets
  that lead to a yes-answer to the $(i+1)$-th query
  and $M_i^{\mathrm{no}}:=M_i\setminus M_i^{\mathrm{yes}}$
  the set of secrets that lead to a no-answer.
  The strategy of codemaker in round $i+1$ is as follows:
  \begin{itemize}
  \item If $|M_i^{\mathrm{yes}}|\geq|M_i^{\mathrm{no}}|$,
    pick a secret from $M_i^{\mathrm{yes}}$
    (and give the answer yes)
  \item Otherwise pick a secret from $M_i^{\mathrm{no}}$
    (and give the answer no)
  \end{itemize}
  By using this strategy, codemaker achieves for every round $i$ that
  \[
    |M_i|=|M_i^{\mathrm{yes}}|+|M_i^{\mathrm{no}}|
     \leq 2\max(|M_i^{\mathrm{yes}}|,|M_i^{\mathrm{no}}|)=2|M_{i+1}|.
  \]
  This implies
  $
    |M_i|\geq 2^{-i}|M_0|
  $. So, for any $i<\log_2(|M_0|)$ we have
  \[
    |M_i|>2^{-\log_2(|M_0|)}|M_0|=\frac{|M_0|}{|M_0|}=1,
  \]
  which means that there are still at least two possible secrets left.
  Since
  \[
    \log_2(|M_0|)=\log_2\left(\prod_{j=\ell}^{k}j\right)=\sum_{j=\ell}^{k}\log_2 j,
  \]
  we obtain the claimed lower bound.
\end{proof}
\begin{corollary}
  Every strategy for Yes-No Permutation-Mastermind (the case $k=n$)
  needs at least
  \[
    \sum\limits_{j=1}^{n}\log_2 j\ge n\log_2 n - n
  \]
  queries in the worst case.
\end{corollary}

\subsection{Upper Bound on the Number of Queries}\label{sec:ub}
\begin{theorem}\label{theo:2}
  Let $k,n\in\mathbb{N}$, $k\ge n$ and $\ell:=k+1-n$.
  For $k=n$, there is a strategy for Yes-No AB-Mastermind that
  identifies every secret code in at most
  $(n-3)\log_2 n+\frac{5}{2}n-1$ queries
  and for $k>n$, there is a strategy
  that identifies every secret code in less than
  $(n-2)\log_2 n+k+1$ queries.
\end{theorem}
\begin{corollary}
The exact asymptotic query complexity of Yes-No Permutation-Mastermind
is $\Theta(n\log_2 n)$.
\end{corollary}
The proof of Theorem~\ref{theo:2} resembles the proof of 
a corresponding result concerning Black-Peg AB-Mastermind~\cite{EGSS18},
except that the information whether a given query
contains a \emph{correct but unidentified} position
is not derived directly but requires special querying
outlined by Algorithm~\ref{infoP} below.
In a nutshell (summarizing with regard to both cases $k=n$ and $k>n$),
the strategy consists of $k$ initial queries
which are the first $n$ positions of shifted versions
of the vector ${(j)}_{j\in[k]}$.
From the answers of the initial queries,
we will be able to learn the secret code position-wise,
keeping record about the positions that have already been identified.
As long as there are consecutive initial queries $a$ and $b$
with the property that $a$ coincides with the secret code in at least one
yet unidentified position but $b$ does not,
we can apply a binary search for the next unidentified
position in $a$, using $\mathcal{O}(\log_2 n)$ further queries.
Such initial queries $a$ and $b$ exist ever after one
(usually after zero) but not all positions of the secret code
have been identified.
\begin{proof}[Proof of Theorem~\ref{theo:2}]
\textbf{The case $k=n$:}
We give a constructive strategy that identifies the positions
of the secret code $y\in {[n]}^n$ one-by-one.
In order to keep record about identified positions of the secret code
we deal with a partial solution vector $x$ that satisfies $x_i\in\{0,y_i\}$
for all $i\in[n]$.
We call the non-zero positions of $x$ \emph{fixed}
and the zero-positions of $x$ \emph{open}.
The fixed positions of $x$
are the identified positions of the secret code.
Remember, that for a query $\sigma=(\sigma_1,\dots,\sigma_n)$
we denote by
\[
		\mathrm{info}(\sigma,y):=\begin{cases}
			\mathrm{yes} & \text{if}\;\; \{i\in[n]\,|\,\sigma_i=y_i\}\ne\emptyset\\
				\mathrm{no} & \text{otherwise}
		\end{cases}
\]
the information
if there is some position in which $\sigma$ coincides with the secret code $y$.
For Yes-No AB-Mastermind the related information whether a query $\sigma$
contains a \emph{correct but unidentified} position
can not always be derived directly
but must be obtained by guessing one or two modifications of $\sigma$,
rearranging those positions that coincide with the partial solution $x$.
The required query procedure is summarized as Algorithm~\ref{infoP}.
\begin{algorithm}[h]
\SetKwInOut{Input}{input}\SetKwInOut{Output}{output}
\Input{Query $\sigma$, partial solution $x$ and secret code $y$}
\Output{Information whether $\sigma$ contains a correct unidentified position}
\lIf{$\sigma$ and $x$ do not coincide}
{$\mathrm{answer} := \mathrm{info}( \sigma, y )$}
\ElseIf{$\sigma$ and $x$ coincide in more than one position}
{
    Let $I\subseteq [n]$ be the set of indices where $\sigma$ and $x$ coincide\;
    Let $\pi:I\rightarrow I$ be a derangement (a permutation without any fixed position)\;
    Obtain guess $\rho$ from $\sigma$ by replacing $\sigma_i$ with $\sigma_{\pi(i)}$ for all $i\in I$\;
    $\mathrm{answer} := \mathrm{info}( \rho, y )$\;
}
\Else%
{
    Let $i$ be the unique index with $\sigma_i=x_i$\;
	\If{$x$ has more then one non-zero position}
    {
        Let $j\ne i$ be another index with $x_j\ne 0$\;
        Obtain guess $\rho$ from $\sigma$ by swapping positions $i$ and $j$\;
        $\mathrm{answer} := \mathrm{info}( \rho, y )$\;
    }
    \Else%
    {
	Choose $j_1\ne i$ and $j_2\ne i$ with $j_1\ne j_2$\;
        Obtain guess $\rho_1$ from $\sigma$ by swapping positions $i$ and $j_1$\;
        Obtain guess $\rho_2$ from $\sigma$ by swapping positions $i$ and $j_2$\;
	\lIf{$\mathrm{info}( \rho_1, y ) = \mathrm{info}( \rho_2, y ) = \mathrm{no}$}{$\mathrm{answer} := \mathrm{no}$}
        \lElse{$\mathrm{answer} := \mathrm{yes}$}
    }
}
\textbf{return} $\mathrm{answer}$\;
\caption{Function infoP}\label{infoP}
\end{algorithm}

\begin{example}
Figure~\ref{fig:1} illustrates the four distinct cases that are considered
by $\mathrm{infoP}$.
In the first and easiest case (panel (a)) the actual query $\sigma$
does not coincide with the partial solution $x$.
Thus, $\sigma$ contains a correct unidentified position
if and only if it contains a correct position at all,
i.e., $\mathrm{infoP}(\sigma,x,y)=\mathrm{info}(\sigma,y)$.
In the second case (panel (b)),
$\sigma$ and $x$ coincide in more than one position,
namely the positions with colors $3$, $9$ and $10$.
The modified query $\rho$ is obtained from $\sigma$ by rearranging
these positions in a way that all identified positions
get a wrong color
while leaving all open positions of $\sigma$ unchanged.
This implies that $\mathrm{infoP}(\sigma,x,y)=\mathrm{info}(\rho,y)$.
Panels (c) and (d) deal with the case that $\sigma$ and $x$ coincide
in exactly one position, say $i$.
If $x$ already contains a further non-zero position $j$,
we obtain $\rho$ from $\sigma$ by swapping positions
$i$ and $j$ in $\sigma$ (the positions with colors $3$ and $5$ in panel (c)).
Again, we obtain that $\mathrm{infoP}(\sigma,x,y)=\mathrm{info}(\rho,y)$.
Finally, if position $i$ is the only yet identified position
of the secret code we have to ask two different modified queries
to derive $\mathrm{infoP}(\sigma,x,y)$ (panel (d)).
We obtain the two queries $\rho_1$ and $\rho_2$,
each by swapping the identified position (here $3$)
with another position in $\sigma$, (here with $1$ and $2$, respectively).
While the color of the identified position is wrong in both modifications
$\rho_1$ and $\rho_2$, 
every other position of $\sigma$ coincides with the corresponding position
of at least one modification.
Therefore, $\mathrm{infoP}(\sigma,x,y)=\mathrm{no}$ if and only if
$\mathrm{info}(\rho_1,y)=\mathrm{info}(\rho_2,y)=\mathrm{no}$.
\begin{figure}
\begin{center}
\begin{tikzpicture}
\draw(-1.4,-0.25)node{(a)};
\draw(-0.2,0.0)node{$x$};
\draw(1.0,0.0)node{$\bullet$}; \draw(2.0,0.0)node{$\bullet$}; \draw(3.0,0.0)node{3}; \draw(4.0,0.0)node{$\bullet$}; \draw(5.0,0.0)node{$\bullet$}; \draw(6.0,0.0)node{6}; \draw(7.0,0.0)node{7}; \draw(8.0,0.0)node{$\bullet$}; \draw(9.0,0.0)node{9}; \draw(10.0,0.0)node{10};
\draw(-0.2,-0.5)node{$\sigma$};
\draw(1.0,-0.5)node{10}; \draw(2.0,-0.5)node{1}; \draw(3.0,-0.5)node{2}; \draw(4.0,-0.5)node{3}; \draw(5.0,-0.5)node{4}; \draw(6.0,-0.5)node{5}; \draw(7.0,-0.5)node{6}; \draw(8.0,-0.5)node{7}; \draw(9.0,-0.5)node{8}; \draw(10.0,-0.5)node{9};
\draw(-1.4,-2.0)node{(b)};
\draw(-0.2,-1.5)node{$x$};
\draw(1.0,-1.5)node{$\bullet$}; \draw(2.0,-1.5)node{$\bullet$}; \draw(3.0,-1.5)node{3}; \draw(4.0,-1.5)node{$\bullet$}; \draw(5.0,-1.5)node{$\bullet$}; \draw(6.0,-1.5)node{6}; \draw(7.0,-1.5)node{7}; \draw(8.0,-1.5)node{$\bullet$}; \draw(9.0,-1.5)node{9}; \draw(10.0,-1.5)node{10};
\draw(-0.2,-2.0)node{$\sigma$};
\draw(1.0,-2.0)node{1}; \draw(2.0,-2.0)node{2}; \draw(3.0,-2.0)node{3}; \draw(4.0,-2.0)node{4}; \draw(5.0,-2.0)node{5}; \draw(6.0,-2.0)node{7}; \draw(7.0,-2.0)node{6}; \draw(8.0,-2.0)node{8}; \draw(9.0,-2.0)node{9}; \draw(10.0,-2.0)node{10};
\draw(-0.2,-2.5)node{$\rho$};
\draw(3.0,-2.5)node[rectangle,draw=black,fill=black!10,minimum height=4.5mm, minimum width=9mm]{};
\draw(9.5,-2.5)node[rectangle,draw=black,fill=black!10,minimum height=4.5mm, minimum width=19mm]{};
\draw(1.0,-2.5)node{1}; \draw(2.0,-2.5)node{2}; \draw(3.0,-2.5)node{10}; \draw(4.0,-2.5)node{4}; \draw(5.0,-2.5)node{5}; \draw(6.0,-2.5)node{7}; \draw(7.0,-2.5)node{6}; \draw(8.0,-2.5)node{8}; \draw(9.0,-2.5)node{3}; \draw(10.0,-2.5)node{9};
\draw(-1.4,-4.0)node{(c)};
\draw(-0.2,-3.5)node{$x$};
\draw(1.0,-3.5)node{$\bullet$}; \draw(2.0,-3.5)node{$\bullet$}; \draw(3.0,-3.5)node{3}; \draw(4.0,-3.5)node{$\bullet$}; \draw(5.0,-3.5)node{$\bullet$}; \draw(6.0,-3.5)node{6}; \draw(7.0,-3.5)node{7}; \draw(8.0,-3.5)node{$\bullet$}; \draw(9.0,-3.5)node{9}; \draw(10.0,-3.5)node{10};
\draw(-0.2,-4.0)node{$\sigma$};
\draw(1.0,-4.0)node{1}; \draw(2.0,-4.0)node{2}; \draw(3.0,-4.0)node{3}; \draw(4.0,-4.0)node{10}; \draw(5.0,-4.0)node{4}; \draw(6.0,-4.0)node{5}; \draw(7.0,-4.0)node{6}; \draw(8.0,-4.0)node{7}; \draw(9.0,-4.0)node{8}; \draw(10.0,-4.0)node{9};
\draw(-0.2,-4.5)node{$\rho$};
\draw(3.0,-4.5)node[rectangle,draw=black,fill=black!10,minimum height=4.5mm, minimum width=9mm]{};
\draw(6.0,-4.5)node[rectangle,draw=black,fill=black!10,minimum height=4.5mm, minimum width=9mm]{};
\draw(1.0,-4.5)node{1}; \draw(2.0,-4.5)node{2}; \draw(3.0,-4.5)node{5}; \draw(4.0,-4.5)node{10}; \draw(5.0,-4.5)node{4}; \draw(6.0,-4.5)node{3}; \draw(7.0,-4.5)node{6}; \draw(8.0,-4.5)node{7}; \draw(9.0,-4.5)node{8}; \draw(10.0,-4.5)node{9};
\draw(-1.4,-6.25)node{(d)};
\draw(-0.2,-5.5)node{$x$};
\draw(1.0,-5.5)node{$\bullet$}; \draw(2.0,-5.5)node{$\bullet$}; \draw(3.0,-5.5)node{3}; \draw(4.0,-5.5)node{$\bullet$}; \draw(5.0,-5.5)node{$\bullet$}; \draw(6.0,-5.5)node{$\bullet$}; \draw(7.0,-5.5)node{$\bullet$}; \draw(8.0,-5.5)node{$\bullet$}; \draw(9.0,-5.5)node{$\bullet$}; \draw(10.0,-5.5)node{$\bullet$};
\draw(-0.2,-6.0)node{$\sigma$};
\draw(1.0,-6.0)node{1}; \draw(2.0,-6.0)node{2}; \draw(3.0,-6.0)node{3}; \draw(4.0,-6.0)node{4}; \draw(5.0,-6.0)node{5}; \draw(6.0,-6.0)node{6}; \draw(7.0,-6.0)node{7}; \draw(8.0,-6.0)node{8}; \draw(9.0,-6.0)node{9}; \draw(10.0,-6.0)node{10};
\draw(-0.2,-6.5)node{$\rho_1$};
\draw(1.0,-6.5)node[rectangle,draw=black,fill=black!10,minimum height=4.5mm, minimum width=9mm]{};
\draw(3.0,-6.5)node[rectangle,draw=black,fill=black!10,minimum height=4.5mm, minimum width=9mm]{};
\draw(1.0,-6.5)node{3}; \draw(2.0,-6.5)node{2}; \draw(3.0,-6.5)node{1}; \draw(4.0,-6.5)node{4}; \draw(5.0,-6.5)node{5}; \draw(6.0,-6.5)node{6}; \draw(7.0,-6.5)node{7}; \draw(8.0,-6.5)node{8}; \draw(9.0,-6.5)node{9}; \draw(10.0,-6.5)node{10};
\draw(-0.2,-7.0)node{$\rho_2$};
\draw(2.5,-7.0)node[rectangle,draw=black,fill=black!10,minimum height=4.5mm, minimum width=19mm]{};
\draw(1.0,-7.0)node{1}; \draw(2.0,-7.0)node{3}; \draw(3.0,-7.0)node{2}; \draw(4.0,-7.0)node{4}; \draw(5.0,-7.0)node{5}; \draw(6.0,-7.0)node{6}; \draw(7.0,-7.0)node{7}; \draw(8.0,-7.0)node{8}; \draw(9.0,-7.0)node{9}; \draw(10.0,-7.0)node{10};
\end{tikzpicture}
\end{center}
\caption{Illustrating the cases considered by $\mathrm{infoP}$.
Panel (a): query $\sigma$ does not coincide with the partial solution $x$;
$\mathrm{infoP}(\sigma,x,y)=\mathrm{info}(\sigma,y)$.
Panel (b): $\sigma$ and $x$ coincide in more than one position;
$\rho$ rearranges these positions of $\sigma$;
$\mathrm{infoP}(\sigma,x,y)=\mathrm{info}(\rho,y)$.
Panel (c): $\sigma$ and $x$ coincide in exactly one position $i$,
but more positions are identified already;
$\rho$ is obtained from $\sigma$ by swapping position $i$
with another identified position $j$;
$\mathrm{infoP}(\sigma,x,y)=\mathrm{info}(\rho,y)$.
Panel (d): Exactly one position is identified and
appears to be correct in $\sigma$;
two modified queries $\rho_1$ and $\rho_2$ must be defined,
each by swapping the identified position with another one;
$\mathrm{infoP}(\sigma,x,y)=\mathrm{no}$ if and only if
$\mathrm{info}(\rho_1,y)=\mathrm{info}(\rho_2,y)=\mathrm{no}$.
\label{fig:1}}
\end{figure}
\end{example}

The codebreaker strategy that identifies the secret code $y$ has two phases.
In the first phase codebreaker guesses an initial sequence
of $n$ queries that has a predefined structure.
In the second phase, the structure of the initial sequence
and the corresponding information by codemaker
enable us to identify correct \emph{positions} $y_i$ of the secret code
one after another, each by using a binary search.
We denote the vector $x$ restricted to the set
$\{s,\dots,\ell\}$ with ${(x_i)}_{i=s}^{\ell}$, $s,\ell\in [n]$.\par
\textbf{\bf Phase 1} Consider the $n$ queries, $\sigma^1,\dots,\sigma^n$,
that are defined as follows: $\sigma^1$ represents the identity map
and for $j\in[n-1]$, we obtain $\sigma^{j+1}$ from $\sigma^{j}$
by a circular shift to the right.
For example, if $n=4$, we have $\sigma^1=(1,2,3,4)$, $\sigma^2=(4,1,2,3)$,
$\sigma^3=(3,4,1,2)$ and $\sigma^4=(2,3,4,1)$.
The codebreaker guesses $\sigma^1,\dots,\sigma^{n}$.\par

\textbf{Phase 2.} Now,
codebreaker identifies the values of $y$ one after another,
using a binary search procedure, that we call findNext.
The idea is to exploit the information that for $1\leq i,j\leq n-1$
we have $\sigma^{j}_{i}=\sigma^{j+1}_{i+1}$,
$\sigma^{n}_{i}=\sigma^{1}_{i+1}$, $\sigma^{j}_{n}=\sigma^{j+1}_{1}$
and $\sigma^{n}_{n}=\sigma^{1}_{1}$.
findNext is used to identify the second correct position
to the last correct position in the main loop of the algorithm.

After the first position of $y$ has been found and fixed in $x$,
there exists a $j\in [n]$ such that $\mathrm{infoP}(\sigma^j,x,y)=\mathrm{no}$.
As long as we have open positions in $x$, we can either find a $j\in[n-1]$
with $\mathrm{infoP}(\sigma^j,x,y)=\mathrm{yes}$
but $\mathrm{infoP}(\sigma^{j+1},x,y)=\mathrm{no}$ 
and set $r:=j+1$, or we have $\mathrm{infoP}(\sigma^n,x,y)=\mathrm{yes}$
but $\mathrm{infoP}(\sigma^1,x,y)=\mathrm{no}$ and set $j:=n$ and $r:=1$.
We call such an index $j$ an \emph{active} index.
Let $j$ be an active index and $r$ its related index.
Let $c$ be the color of some position of $y$ that is already identified
and fixed in the partial solution $x$.
With $\ell_{j}$ and $\ell_{r}$ we denote the position of color $c$
in $\sigma^{j}$ and $\sigma^{r}$, respectively.
The color $c$ serves as a pivot color for identifying
a correct position $m$ in $\sigma^j$ that is not fixed, yet.
There are two possible modes for the binary search
that depend on the fact if $m\le \ell_j$.
The mode is indicated by a Boolean variable $\mathrm{leftS}$
and determined by lines 5--9 of findNext.
Clearly, $m\le \ell_j$ if $\ell_j=n$.
Otherwise, codebreaker guesses
\[
	\sigma^{j,0}:=\left(c,{(\sigma^{j}_i)}_{i=1}^{\ell_{j}-1},{(\sigma^{j}_i)}_{i=\ell_{j}+1}^{n}\right).
\]
By the information $\sigma^{j}_{i}=\sigma^{r}_{i+1}$
we obtain that
${(\sigma^{j}_i)}_{i=1}^{\ell_{j}-1}\equiv {(\sigma^{r}_{i})}_{i=2}^{\ell_{j}}$.
We further know that every open color has a wrong position in $\sigma^r$.
For that reason, $\mathrm{infoP}(\sigma^{j,0},x,y)=\mathrm{no}$
implies that $m\le \ell_j$.
\begin{algorithm}[h]
\SetKwInOut{Input}{input}\SetKwInOut{Output}{output}
\Input{Secret code $y$, partial solution $x\ne 0$ and an active index $j\in[n]$}
\Output{A correct open position in $\sigma^j$}
\lIf{$j=n$}{$r:=1$}\lElse{$r:=j+1$}
Choose the color $c$ of some non-zero position of $x$\;
Let $\ell_j$ and $\ell_r$ be the positions with color $c$ in $\sigma^j$ and $\sigma^r$, respectively\;
\lIf{$\ell_j=n$}{$\mathrm{leftS:=true}$}
\Else%
{
	$\sigma^{j,0}:=\left(c,{(\sigma^{j}_{i})}_{i=1}^{\ell_j-1},{(\sigma^{j}_{i})}_{i=\ell_j+1}^{n}  \right)$\;
	\lIf{$\mathrm{infoP}(\sigma^{j,0},x,y)$}{$\mathrm{leftS:=false}$}
	\lElse{$\mathrm{leftS:=true}$}
}
\If{$\mathrm{leftS}$}
{
    $a:=1$\;
    $b:=\ell_j$\;
}
\Else%
{
    $a:=\ell_r$\;
    $b:=n$\;
}
\While{$b>a$}
{	
$\ell:=\lceil \frac{a+b}{2} \rceil$\tcp*{position for color $c$}
\lIf{$\mathrm{leftS}$}{$\sigma^{j,\ell}:=\left({(\sigma^{j}_{i})}_{i=1}^{\ell-1},c,{(\sigma^{r}_{i})}_{i=l+1}^{\ell_j},{(\sigma^{j}_{i})}_{i=\ell_j+1}^{n}\right)$}
\lElse{$\sigma^{j,\ell}:=\left({(\sigma^{r}_{i})}_{i=1}^{\ell_r-1},{(\sigma^{j}_{i})}_{i=\ell_r}^{\ell-1},c,{(\sigma^{r}_{i})}_{i=l+1}^{n}\right)$}
	\lIf{$\mathrm{infoP}(\sigma^{j,\ell},x,y)$}{$b:=l-1$}
	\lElse{$a:=l$}
}
\textbf{return} $b$\;
\caption{Function findNext}\label{findNext}
\end{algorithm}
The binary search for the exact value of $m$ is done in the interval $[a,b]$,
where $m$ is initialized as $n$ and $[a,b]$ as
\[
[a,b]:=\begin{cases}
[1,\ell_j] & \text{if}\;\,\mathrm{leftS}\\
[\ell_r,n] & \text{else}
\end{cases}
\]
(lines 10--15 of findNext).
In order to determine if there is an open correct position
on the left side of the current center $\ell$ of $[a,b]$ in $\sigma^j$
we can define a case dependent query:
\[
\sigma^{j,\ell}:=\begin{cases}
	\left({(\sigma^{j}_{i})}_{i=1}^{\ell-1},c,{(\sigma^{r}_{i})}_{i=l+1}^{\ell_j},{(\sigma^{j}_{i})}_{i=\ell_j+1}^{n}\right) & \hspace{-3mm}\text{if}\;\,\mathrm{leftS}\\
	\left({(\sigma^{r}_{i})}_{i=1}^{\ell_r-1},{(\sigma^{j}_{i})}_{i=\ell_r}^{\ell-1},c,{(\sigma^{r}_{i})}_{i=l+1}^{n}\right) & \,\text{ else}
\end{cases}
\]
In the first case, the first $\ell-1$ positions of $\sigma^{j,\ell}$
coincide with those of $\sigma^j$.
The remaining positions of $\sigma^{j,\ell}$
cannot coincide with the corresponding positions of the secret code
if they have not been fixed, yet.
This is because the $\ell$-th position of $\sigma^{j,\ell}$
has the already fixed value $c$,
positions $\ell+1$ to $\ell_j$
coincide with the corresponding positions of $\sigma^r$
which satisfies $\mathrm{infoP}(\sigma^r,x,y)=\mathrm{no}$
and the remaining positions have been checked to be wrong in this case
(cf.\ former definition of $\mathrm{leftS}$ in line 5 and line 9, respectively).
Thus, there is a correct open position on the left side of $\ell$ in $\sigma^j$,
if and only if $\mathrm{infoP}(\sigma^{j,\ell},x,y)=\mathrm{yes}$.
In the second case, the same holds for similar arguments.
Now, if there is a correct open position to the left of $\ell$,
we update the binary search interval $[a,b]$ by $[a,\ell-1]$.
Otherwise, we update $[a,b]$ by $[\ell,b]$.
\begin{example}
Suppose, that for $n=10$ the secret code $y$ and the partial solution $x$
are given as in the top panel of Figure~\ref{fig:2}
and that we have first identified the position with color $1$,
such that $1$ is our pivot color.
The initial $10$ queries $\sigma^1,\dots,\sigma^{10}$ 
together with their current $\mathrm{infoP}$ measures
are given in the mid panel of Figure~\ref{fig:2}.
We see, that the highlighted queries, $\sigma^4$ and $\sigma^5$,
can be used for the binary search with findNext,
since $\sigma^4$ has a correct not yet identified position
but $\sigma^5$ has not.
So the active indices are $j=4$ and $r=5$ and the corresponding
pivot color positions in $\sigma^4$ and $\sigma^5$ are
$\ell_j=4$ and $\ell_r=5$.
The first query of findNext (cf.\ lower panel of Figure~\ref{fig:2})
is $\sigma^a$.
It begins with the pivot color,
followed by the first $3$ positions of $\sigma^4$
(positions $2$ to $4$ of $\sigma^5$)
and positions $5$ to $10$ of $\sigma^4$ (cf.\ line 7 of findNext).
Since $\mathrm{infoP}(\sigma^a,x,y)=\mathrm{yes}$,
the left most correct but unidentified position in $\sigma^4$
is none of its first $4$ positions.
Thus, the binary search is continued in the interval $[5,10]$.
It is realized by queries $\sigma^b$, $\sigma^c$, and $\sigma^d$,
which are composed according to line 20 of findNext (in this case),
and finally identifies position $8$ with color $5$ of the secret code
(generally the position left to the left most pivot color position
that receives the answer \enquote{yes} in the binary search).
\begin{figure}
\begin{center}
\begin{tikzpicture}
\draw(-1.4,0.25) node{(a)};
\draw(-0.2,0.5) node{$y$};
\draw(-0.2,0.0) node{$x$};
\draw(1.0,0.5)node{9}; \draw(2.0,0.5)node{10}; \draw(3.0,0.5)node{6}; \draw(4.0,0.5)node{8}; \draw(5.0,0.5)node{4}; \draw(6.0,0.5)node{2}; \draw(7.0,0.5)node{7}; \draw(8.0,0.5)node{5}; \draw(9.0,0.5)node{1}; \draw(10.0,0.5)node{3};
\draw(1.0,0.0)node{$\bullet$}; \draw(2.0,0.0)node{$\bullet$}; \draw(3.0,0.0)node{6}; \draw(4.0,0.0)node{8}; \draw(5.0,0.0)node{$\bullet$}; \draw(6.0,0.0)node{2}; \draw(7.0,0.0)node{$\bullet$}; \draw(8.0,0.0)node{$\bullet$}; \draw(9.0,0.0)node{1}; \draw(10.0,0.0)node{3};
\draw(-1.4,-3.5) node{(b)};
\draw(5.5,-3.0)node[rectangle,draw=black,fill=black!30,minimum height=4.5mm, minimum width=100mm]{};
\draw(5.5,-3.5)node[rectangle,draw=black,fill=black!10,minimum height=4.5mm, minimum width=100mm]{};
\foreach \y in {1,2,...,10}
{
	\draw(-0.2,-1-0.5*\y) node{$\sigma^{\y}$};
	\foreach \x in {1,2,...,10}
	{
		\pgfmathtruncatemacro{\z}{mod(int(10+\x-\y),10)+1}%
		\draw(\x,-1-0.5*\y) node{\z};
	}
}
\draw(11.2,-1.5) node{yes};
\draw(11.2,-2.0) node{yes};
\draw(11.2,-2.5) node{yes};
\draw(11.2,-3.0) node{yes};
\draw(11.2,-3.5) node{no};
\draw(11.2,-4.0) node{no};
\draw(11.2,-4.5) node{no};
\draw(11.2,-5.0) node{no};
\draw(11.2,-5.5) node{no};
\draw(11.2,-6.0) node{no};
\draw(5.5,-1.0) node{initial queries};
\draw(11.2,-1.0) node{$\mathrm{infoP}$};
\draw(-1.4,-8.0) node{(c)};
\draw(-0.2,-7.5) node{$\sigma^a$};
\draw(-0.2,-8.0) node{$\sigma^b$};
\draw(-0.2,-8.5) node{$\sigma^c$};
\draw(-0.2,-9.0) node{$\sigma^d$};
\draw(3.0,-7.5)node[rectangle,draw=black,fill=black!10,minimum height=4.5mm, minimum width=29mm]{};
\draw(7.5,-7.5)node[rectangle,draw=black,fill=black!30,minimum height=4.5mm, minimum width=59mm]{};
\draw(1.0,-7.5)node{1}; \draw(2.0,-7.5)node{8}; \draw(3.0,-7.5)node{9}; \draw(4.0,-7.5)node{10}; \draw(5.0,-7.5)node{2}; \draw(6.0,-7.5)node{3}; \draw(7.0,-7.5)node{4}; \draw(8.0,-7.5)node{5}; \draw(9.0,-7.5)node{6}; \draw(10.0,-7.5)node{7};
\draw(2.5,-8.0)node[rectangle,draw=black,fill=black!10,minimum height=4.5mm, minimum width=39mm]{};
\draw(5.5,-8.0)node[rectangle,draw=black,fill=black!30,minimum height=4.5mm, minimum width=19mm]{};
\draw(9.0,-8.0)node[rectangle,draw=black,fill=black!10,minimum height=4.5mm, minimum width=29mm]{};
\draw(1.0,-8.0)node{7}; \draw(2.0,-8.0)node{8}; \draw(3.0,-8.0)node{9}; \draw(4.0,-8.0)node{10}; \draw(5.0,-8.0)node{2}; \draw(6.0,-8.0)node{3}; \draw(7.0,-8.0)node{1}; \draw(8.0,-8.0)node{4}; \draw(9.0,-8.0)node{5}; \draw(10.0,-8.0)node{6};
\draw(2.5,-8.5)node[rectangle,draw=black,fill=black!10,minimum height=4.5mm, minimum width=39mm]{};
\draw(6.0,-8.5)node[rectangle,draw=black,fill=black!30,minimum height=4.5mm, minimum width=29mm]{};
\draw(9.5,0-8.5)node[rectangle,draw=black,fill=black!10,minimum height=4.5mm, minimum width=19mm]{};
\draw(1.0,-8.5)node{7}; \draw(2.0,-8.5)node{8}; \draw(3.0,-8.5)node{9}; \draw(4.0,-8.5)node{10}; \draw(5.0,-8.5)node{2}; \draw(6.0,-8.5)node{3}; \draw(7.0,-8.5)node{4}; \draw(8.0,-8.5)node{1}; \draw(9.0,-8.5)node{5}; \draw(10.0,-8.5)node{6};
\draw(2.5,-9.0)node[rectangle,draw=black,fill=black!10,minimum height=4.5mm, minimum width=39mm]{};
\draw(6.5,-9.0)node[rectangle,draw=black,fill=black!30,minimum height=4.5mm, minimum width=39mm]{};
\draw(10.0,-9.0)node[rectangle,draw=black,fill=black!10,minimum height=4.5mm, minimum width=9mm]{};
\draw(1.0,-9.0)node{7}; \draw(2.0,-9.0)node{8}; \draw(3.0,-9.0)node{9}; \draw(4.0,-9.0)node{10}; \draw(5.0,-9.0)node{2}; \draw(6.0,-9.0)node{3}; \draw(7.0,-9.0)node{4}; \draw(8.0,-9.0)node{5}; \draw(9.0,-9.0)node{1}; \draw(10.0,-9.0)node{6};
\draw(11.2,-7.5) node{yes};
\draw(11.2,-8.0) node{no};
\draw(11.2,-8.5) node{no};
\draw(11.2,-9.0) node{yes};
\draw(5.5,-7.0) node{queries of $\mathrm{findNext}$};
\draw(11.2,-7.0) node{infoP};
\end{tikzpicture}
\end{center}
\caption{Panel (a): secret code $y$ and partial solution vector $x$.
Panel (b): the initial queries $\sigma^j$
and their responses $\mathrm{infoP}(\sigma^j,x,y)$,
indicating if a query and the secret code coincide in any position 
that has not been identified, yet (i.e., in any $0$-position of $x$).
Panel (c): binary search queries to identify the next secret position.
The highlighted subsequences correspond to the subsequences
of the initial queries that have been selected to apply the binary search.
\label{fig:2}}
\end{figure}
\end{example}

{\bf The Main Algorithm.} The main algorithm is outlined as Algorithm~\ref{findAll}.
\begin{algorithm}
Initialize $x:=(0,0,\dots,0)$\;
Guess the queries $\sigma^i$, $i\in[n-1]$\;
Initialize $v\in{\{\mathrm{yes,no}\}}^n$ by $v_i:=\mathrm{info}(\sigma^i,y)$, $i\in[n]$\;
\If{$v_i=\mathrm{yes}$ $\forall i\in[n]$}
{
	Find position $m$ with a correct color in $\sigma^1$ by at most $\lfloor\frac{n}{2}\rfloor+1$ queries\;
	$x_m:=\sigma^1_m$\;
	$v_1:=\mathrm{no}$\;
}
\While{$|\{i\in[n]\,|\,x_i=0\}|>2$}
{
	Use $v$ to choose an active index $j\in[n]$\tcp*{($v_j=\mathrm{yes}$, $v_{j+1}=\mathrm{no}$)}
	$m:=\mathrm{findNext}(x,y,j)$\;
	$x_m:=\sigma^{j}_m$\;
	$v_j:=\mathrm{infoP}(\sigma^j,x,y)$\;
}
Make at most two more queries to find the remaining two unidentified colors\;
\caption{Codebreaker Strategy for Permutations}\label{findAll}
\end{algorithm}
It starts with an empty partial solution
and finds the positions of the secret code $y$ one-by-one.
The vector $v$ keeps record about which of the initial queries $\sigma^1,\dots,\sigma^n$ 
coincides with the secret code $y$ in some open position.
Thus, $v$ is initialized by $v_i:=\mathrm{info}(\sigma^i,y)$, $i\in[n]$.
The main loop always requires an active index.
For that reason, if $v_i=\mathrm{yes}$ for all $i\in[n]$ in the beginning,
we first identify the correct position in $\sigma^1$
(which is unique in this case)
by $\lfloor\frac{n}{2}\rfloor+1$ queries
(each swapping two positions of $\sigma^1$)
and update $x$ and $v$, correspondingly.
After this step, there will always exist an active index.
Every call of findNext in the main loop augments $x$
by a correct solution value.
One call of findNext
requires at most $1+\lceil\log_2 n\rceil$ queries
if the partial solution $x$ contains more than one non-zero position,
and at most $2+2\lceil\log_2 n\rceil$ queries
(two queries for each call of infoP) if $x$ has exactly one non-zero position.
Thus, Algorithm~\ref{findAll} does not need more than
$(n-2)\lceil\log_2 n\rceil+\frac{5}{2}n-1$ queries to break the secret code
inclusive the $n-1$ initial queries,
$\lfloor\frac{n}{2}\rfloor+1$ queries to find the first correct position,
$n-3$ calls of findNext and $2$ final queries.\\

\textbf{The case $k>n$:}
Let $y=(y_1,\dots,y_n)$ be the code that must be found.
We use the same notations as above.

\textbf{Phase 1.} Consider the
$k$ queries $\overline{\sigma}^1,\dots,\overline{\sigma}^k$,
where $\overline{\sigma}^1$ represents the identity map on $[k]$
and for $j\in[k-1]$, we obtain $\overline{\sigma}^{j+1}$
from $\overline{\sigma}^{j}$ by a circular shift to the right.
We define $k$ codes $\sigma^1,\dots,\sigma^k$
by $\sigma^j={(\overline{\sigma}^j_i)}_{i=1}^n$, $j\in[k]$.
For example, if $k=5$ and $n=3$, we have $\sigma^1=(1,2,3)$,
$\sigma^2=(5,1,2)$, $\sigma^3=(4,5,1)$, $\sigma^4=(3,4,5)$
and $\sigma^5=(2,3,4)$.
Within those $k$ codes, every color appears exactly once at every position
and, thus, there are at least $k-n$ initial queries
that do not contain any correct position.
Since $k>n$, this implies
\begin{lemma}\label{exists0}
There is a $j\in[k]$ with $\mathrm{info}(\sigma^j,y)=\mathrm{no}$.
\end{lemma}

\textbf{Phase 2.} Having more colors than positions,
we can perform our binary search for a next
correct position without using a pivot color.
The corresponding simplified version of findNext
is outlined as Algorithm~\ref{findNext2}.
\begin{algorithm}
\SetKwInOut{Input}{input}\SetKwInOut{Output}{output}
\Input{Code $y$, partial solution $x\ne 0$ and an active index $j\in[k]$}
\Output{A position $m$ that is correct in $\sigma^j$}
\lIf{$j=n$}{$r:=1$}\lElse{$r:=j+1$}
$a:=1$, $b:=n$\;
\While{$b>a$}
{	
	$\ell:=\lceil \frac{a+b}{2} \rceil$\tcp*{mid position of current interval}
	Guess $\sigma:=\left({(\sigma^{r}_{i})}_{i=1}^{\ell-1},{(\sigma^{j}_{i})}_{i=\ell}^{n}\right)$\;
	$s:=\mathrm{infoP}(\sigma,x,y)$\;
	\lIf{$s=\mathrm{yes}$}{$a:=l$}
	\lElse{$b:=l-1$}
}
\textbf{return} $a$\;
\caption{Function findNext for $k>n$}\label{findNext2}
\end{algorithm}
Using that version of findNext also allows to simplify our main algorithm 
(Algorithm~\ref{findAll}) by adapting lines 2 and 3,
and, due to Lemma~\ref{exists0}, skipping lines 4--7, as findNext can
be already applied to find the first correct position.
Thus, for  the required number of queries to break the secret code we have:
the initial $k-1$ queries,
a call of the modified findNext for every but the last two positions
and one or two final queries.
This yields, that the modified Mastermind Algorithm
breaks the secret code in at most $(n-1)\lceil\log_2 n\rceil+k+1$ queries.
\end{proof}

\section{Conclusions}
We showed that deterministic algorithms
for the identification of a secret code in Black-Peg AB-Mastermind can,
with a slight modification, also be applied to Yes-No AB-Mastermind
and yields upper bounds on its query complexity.
Utilizing a simple codemaker cheating strategy,
we further derive corresponding lower bounds for Yes-No AB-Mastermind.
One challenge of this variant is that codemaker's answers are restricted
to the information whether query and secret code coincide in any position.
A bigger challenge with AB-Mastermind is
that no color repetition is allowed in a query
whereas most strategies for other Mastermind variants
exploit the property of color repetition.
While for most Mastermind variants there is a gap
between lower and upper bounds on the worst case number of queries
to break the secret code,
our results imply that this number is $\Theta(n\log n)$
for the most popular case $k=n$ of Yes-No AB-Mastermind,
which is also referred to as Yes-No Permutation-Mastermind.
To our knowledge, this result is a first exact asymptotic
query complexity proof for a multicolor Mastermind variant,
where both secret code and queries are chosen from the same set,
here $[n]^n$.

A future challenge will be studying the static variant of Yes-No AB-Mastermind
(where codebreaker must give all but one queries
in advance of codemaker's answers).
Lower and upper bounds for static Black-Peg AB-Mastermind were provided
as $\Omega(n\log n)$ and $\mathcal{O}(n^{1.525})$,
respectively~\cite{GJSS17}.

\vspace{3pt} 
\noindent
{\small \textbf{Codeavailability:} We provide Matlab/Octave implementations
of the codebreaker strategy via GitHub, a permanent version of which
is archived in a public zenodo repository~\cite{Sau20}.
}
\vspace{3pt}

\bibliographystyle{abbrv} 
\bibliography{references}

\begin{thebibliography}{10}

\bibitem{AAD+19}
P.~Afshani, M.~Agrawal, B.~Doerr, C.~Doerr, K.~G. Larsen, and K.~Mehlhorn.
\newblock The query complexity of a permutation-based variant of mastermind.
\newblock {\em Discrete Applied Mathematics}, 260:28--50, 2019.

\bibitem{BCS16}
A.~Berger, C.~Chute, and M.~Stone.
\newblock Query {C}omplexity of {M}astermind {V}ariants.
\newblock arXiv:1607.04597, 17 pages, 2016.

\bibitem{C83}
V.~Chv\'atal.
\newblock Mastermind.
\newblock {\em Combinatorica}, 3:325--329, 1983.

\bibitem{DDST16}
B.~Doerr, C.~Doerr, R.~Sp\"{o}hel, and H.~Thomas.
\newblock Playing {M}astermind with {M}any {C}olors.
\newblock {\em Journal of the ACM}, 63(5):1--23, 2016.

\bibitem{EGSS18}
M.~El~Ouali, C.~Glazik, V.~Sauerland, and A.~Srivastav.
\newblock On the {Q}uery {C}omplexity of {B}lack-{P}eg {AB}-{M}astermind.
\newblock {\em Games}, 9(1):2, 2018.

\bibitem{ES13}
M.~El~Ouali and V.~Sauerland.
\newblock Improved {A}pproximation {A}lgorithm for the {N}umber of {Q}ueries
  {N}ecessary to {I}dentify a {P}ermutation.
\newblock In T.~Lecroq and L.~Mouchard, editors, {\em Proceedings of the 24th
  International Workshop on Combinatorial Algorithms, Rouen, France, July 2013
  (IWOCA 2013)}, number 8288 in Lecture Notes in Computer Science, pages
  443--447, 2013.

\bibitem{Sau20}
M.~El~Ouali and V.~Sauerland.
\newblock Github repositry yn-ab-mastermindv1.0: Codebreaker strategies for
  yes-no ab-mastermind (matlab/octave), 2020.

\bibitem{ER63}
P.~Erd{\H{o}}s and C.~R\'enyi.
\newblock On {T}wo {P}roblems in {I}nformation {T}heory.
\newblock {\em Publications of the Mathematical Institute of the Hungarian
  Academy of Science}, 8:229--242, 1963.

\bibitem{GJSS17}
C.~Glazik, G.~J{\"{a}}ger, J.~Schiemann, and A.~Srivastav.
\newblock Bounds for {S}tatic {B}lack-{P}eg {AB} {M}astermind.
\newblock In X.~Gao, H.~Du, and M.~Han, editors, {\em Proceedings of the 11th
  International Conference on Combinatorial Optimization and Applications,
  Shanghai, China, December 16-18 ({COCOA 2017}), Part {II}}, number 10628 in
  Lecture Notes in Computer Science, pages 409--424. Springer, 2017.

\bibitem{Goo09b}
M.~T. Goodrich.
\newblock On the algorithmic complexity of the {M}astermind game with black-peg
  results.
\newblock {\em Information Processing Letters}, 109:675--678, 2009.

\bibitem{JP11}
G.~J\"ager and M.~Peczarski.
\newblock The number of pessimistic guesses in {G}eneralized {B}lack-peg
  {M}astermind.
\newblock {\em Information Processing Letters}, 111:993--940, 2011.

\bibitem{K77}
D.~E. Knuth.
\newblock The computer as a master mind.
\newblock {\em Journal of Recreational Mathematics}, 9:1--5, 1977.

\bibitem{KT86}
K.~Ko and S.~Teng.
\newblock On the {N}umber of {Q}ueries {N}ecessary to {I}dentify a
  {P}ermutation.
\newblock {\em Journal of Algorithms}, 7:449--462, 1986.

\end{thebibliography}
\end{document}